\documentclass[12pt]{iopart}
\usepackage{iopams}

\usepackage{latexsym,amssymb}
\usepackage[english]{babel}
\usepackage{verbatim}
\usepackage{color}

\usepackage{amsthm}

\newcommand{\eqref}[1]{(\ref{#1})}

\newtheorem{lemma}{Lemma}
\newtheorem{remark}{Remark}

\newcommand{\ket}[1]{ | #1 \rangle}
\newcommand{\bra}[1]{ \langle #1 |}

\newcommand{\la}{\lambda}

\def\one{\leavevmode\hbox{\small1\normalsize\kern-.33em1}}

\newcommand{\beq}{\begin{equation}}
\newcommand{\eeq}{\end{equation}}
\newcommand{\bea}{\begin{eqnarray}}
\newcommand{\eea}{\end{eqnarray}}

 \def\lam{\lambda} 
   
\def\rA{\rm A}  \def\rB{\rm B}
\def\bfb{\mathbf b}

\newcommand{\opsi}{\bar{\psi}}

\newcommand{\nn}{\nonumber}

\newtheorem{theorem}{Theorem}
\newtheorem{proposition}{Proposition}
\newtheorem{corollary}{Corrolary}

\begin{document}

\title{Quantum weighted entropy and its properties}


\author{Y.~Suhov$^1$ and S.~Zohren$^2$}

\address{
$^1$ Department of Mathematics, Penn State University, PA 16802, USA; \\ 
DPMMS, University of Cambridge, CB30WB, UK; \\ 
Institute for Information Transmission Problems, RAS, 127994 Moscow, Russia \\
$^2$ Department of Physics, Pontifica Universidade Cat\'olica, Rio de Janeiro, Brazil
}



\begin{abstract}
We introduce quantum weighted entropy in analogy to an earlier notion of (classical) weighted entropy and derive many of its properties. These include the subadditivity, concavity and strong subadditivity property of quantum weighted entropy, as well as an analog of the Araki-Lieb inequality. Interesting byproducts of the proofs are a weighted analog of Klein's inequality and non-negativity of quantum weighted relative entropy. A main difficulty is the fact that the weights in general do not commute with the density matrices.
\end{abstract}

\maketitle


\section{Introduction}
Shannon entropy and its quantum analog, von Neumann entropy, play essential roles in classical and quantum information theory \cite{CT,Information,Nielsen}. There are many generalisations of Shannon entropy, such as R\'enyi entropy, which have been proposed both in the classical as well as quantum case. Another interesting generalisation is (classical) weighted entropy \cite{BG,G}. The idea behind weighted entropy is the incorporation of further characteristics of each event through a weight assigned to it in addition to its probability. Weighted entropy has been used at several places in the information theory and computer science literature (see for instance \cite{ShMM,SiB,K,SrV,S} and references therein) including in machine learning applications \cite{MMN}. However, the natural quantum analog of weighted entropy has apparently not yet been considered in the quantum information theory literature.

The aim of this letter is to introduce \emph{quantum weighted entropy} and to prove several basic properties, including subadditivity, concavity and strong subadditivity of quantum weighted entropy and an analog of the Araki-Lieb inequality. Many of the corresponding trace inequalities for von Neumann entropy have important applications in quantum information theory. As an example let us mention strong subadditivity of standard von Neumann entropy which was conjectured in \cite{LR} and then proven by Lieb and Ruskai \cite{LR1,LR2} (see also \cite{NP} for a modern simplified proof). Amongst many of its applications is the thermodynamic limit of entropy per volume which was already considered in the classical case \cite{RR}. We refer the reader to \cite{Ruskai2002} for a review. Some of the above trace inequalities are quantum analogs of a series of inequalities for classical weighted entropy recently considered in \cite{SY}.

Let us now give a formal definition of quantum weighted entropy. Consider a quantum mechanical system with Hilbert space $\mathcal{H}$ and a density matrix $\rho$ on $\mathcal{H}$. For an Hermitian, positive definite matrix $\phi$ on $\mathcal{H}$, which from here on we simply refer to as \emph{weight}, we define the \emph{quantum weighted entropy} as follows
\beq \label{def1}
S_\phi(\rho)=- \tr (\phi\rho\log\rho).
\eeq
One sees that for $\phi=1_{\mathcal{H}}$, where $1_{\mathcal{H}}$ is the identity matrix on $\mathcal{H}$, the quantum wighted entropy reduces to the standard von Neumann entropy. 

Before moving to the discussion of different properties and their proofs in the next sections let us first comment on potential difficulties. Consider for the moment the well-known Gibbs inequality which yields positivity of von Neumann entropy whose weighted analog we will discuss in the next section. The main difficulty in the proof, as compared to the corresponding classical result for Shannon entropy, is the fact that the different (reduced) density matrices in general do not commute. Similarly, when extending many trace inequalities for von Neumann entropy to quantum weighted entropy the difficulty lies in the fact that now also the weight $\phi$ is not commuting with the density matrices.

\section{Quantum weighted relative entropy and Gibbs inequality}

Extending standard classical notions, we can also introduce the \emph{quantum weighted relative entropy} or weighted Kullback-Leibler divergence as
\beq \label{def2}
D_\phi(\rho \|\sigma  )= \tr (\phi\rho\log\rho)-\tr (\phi\rho\log\sigma).
\eeq
Here and below $\rho$ is a density matrix and $\sigma$ is positive definite in $\mathcal{H}$.

An important property of quantum weighted relative entropy is given by the weighted Klein's inequality
\begin{lemma}[Weighted Klein's inequality] \label{Klein}
Assume that $X, Y, W$ are Hermitian positive definite matrices on a Hilbert space $\mathcal{H}$. Then if $f$ is a convex function one has
\beq
\tr \left(W(f(Y) - f(X))\right) \geq \tr \left( W(Y - X)f'(X)\right) .
\eeq 
In particular for $f(x)=x\log x$ one has
\beq\label{Kleineq2}
\tr \left(W Y (\log Y - \log X) \right) \geq \tr \left( W(Y - X) \right) .
\eeq 
\end{lemma}
The proof is given in the appendix. At this point we only note that a slight difficulty in comparison to the standard Klein's inequality is caused by the fact that $W$ in general does not commute with $X$ and $Y$. An important result which can be easily derived from the weighted Klein's inequality is the weighted Gibbs  inequality. 

\begin{theorem}[Weighted Gibbs inequality] \label{Gibbs}
Under the condition $\tr \,\phi\,\rho\geq \tr \,\phi\,\sigma$ one has
\beq
D_\phi(\rho \|\sigma  )\geq 0, 
\eeq
with equality if and only if $\rho =\sigma$.
\end{theorem}

\begin{proof}
The proof of the weighted Gibbs inequality follows directly from the weighted Klein's inequality, Lemma \ref{Klein}. In particular for $X=\sigma$, $Y=\rho$, $W=\phi$ and $f(x)=x\log x$, the result immediately follows form \eqref{Kleineq2} under the condition $\tr \,\phi\,\rho\geq \tr \,\phi\,\sigma$.
\end{proof}

\begin{remark} Note that the condition $\tr \,\phi\,\rho\geq \tr \,\phi\,\sigma$ is physical and a minimum necessary requirement for the weighted Gibbs inequality. As expected, it is automatically satisfied if $\phi=1_{\mathcal{H}}$, i.e.\ in the case of standard von Neumann entropy.
\end{remark}

\section{Basic properties of quantum weighted entropy}

We now discuss three propositions with useful basic properties of the quantum weighted entropy. 

\begin{proposition} Denote by $\ket{e_1},\ldots ,\ket{e_d}$ the normalised eigenvectors of $\rho$ and by $\lam_1,\ldots ,\lam_d$ the corresponding eigenvalues. 
\begin{enumerate}
\item The quantum weighted entropy $S_\phi (\rho )$ is non-negative and zero if and only if either (i) $\rho$ is pure or (ii) $\langle e_i|\phi| e_i\rangle=0$ whenever $0<\lam_i<1$. 
\item $S_\phi (\rho )=S_{\phi'} (\rho )$ if $\langle e_i|\phi| e_i\rangle
=\langle e_i|\phi'|e_i\rangle$ whenever $0<\lam_i<1$. In this case we say that $\phi$ and $\phi'$ are $\rho$-conjugate.
\end{enumerate}
\end{proposition}

\begin{proof}
The results of this proposition are a consequence of the eigenvalue decomposition of $\rho$ which yields
\beq\label{Srhoeigen}
S_\phi (\rho )=-\sum_i\langle e_i|\phi|e_i\rangle\lam_i\log\,\lam_i.
\eeq 
We directly see from this that $S_\phi (\rho )$ is non-negative. If $\rho$ is pure $\lam_i\log\,\lam_i=0, \forall i$ which implies $S_\phi (\rho )=0$. Otherwise, if $\rho$ is not pure one has $\lam_i<1, \forall i$ and it is clear from the above that in this case $S_\phi (\rho )=0$ if and only if $\langle e_i|\phi| e_i\rangle=0$ whenever $0<\lam_i<1$. This proves the first part of the proposition. The second part of the proposition follows directly from \eqref{Srhoeigen}.
\end{proof}

Above we have already shown that $S_\phi (\rho )\geq0$. The following proposition gives an upper bound on the quantum weighted entropy.

\begin{proposition} Suppose the rank of $\phi$ equals $m\leq d$ and let $P= P_\phi$ be the orthoprojection to the range of $\phi$. If $\tr \phi\rho\geq{\rm{tr}}\,\phi/m$ then
\beq
S_{\phi}(\rho )\leq S_{\phi}(P/m )=\big(\log\,m\big)\,{\rm{tr}}\,\phi 
\eeq
with equality if and only if $\rho = P/m$.
\end{proposition}

\begin{proof}
The proposition is a direct consequence of the quantum weighted Gibbs inequality, i.e.\ Theorem \ref{Gibbs}, with $\sigma=P/m$.
\end{proof}

Let us now discuss properties of the weighted entropy with respect to the operations of
purification and partial traces (see \cite{Nielsen} for a standard textbook account). Purifications are based on the Schmidt decomposition of pure states which itself is a consequence of the singular value decomposition of complex matrices. Let us quickly recall the concept of purifications. If $\rho_A$ is a density matrix on $A$, then there exists a reference system $R$ and a pure state $\ket{\chi}$ on $AR$ such that $\rho_A=\tr_R \ket{\chi}\bra{\chi}$. Denote also $\rho_R=\tr_A \ket{\chi}\bra{\chi}$. Standard arguments show that $\rho_A$ and $\rho_R$ have the same collection $\{\lam_i\}$ of non-negative eigenvalues. Furthermore,  if we denote by $\{\ket{e^{A}_i}\}$ and $\{\ket{e^{R}_i}\}$ the corresponding eigenvectors of $\rho_A$ and $\rho_R$, then one finds
\bea
S_{\phi_1}(\rho_A)&=&-\sum\limits_i\langle e^{A}_i|\phi_1| e^{A}_i\rangle\lam_i\log\,\lam_i , \\
S_{\phi_2}(\rho_R)&=&-\sum\limits_i\langle e^{R}_i|\phi_2| e^{R}_i\rangle\lam_i\log\,\lam_i .
\eea
This proves the following proposition:
\begin{proposition}
Let $\ket{\chi}$ be a pure state  on $AR$ with $\rho_A=\tr_R \ket{\chi}\bra{\chi}$
and $\rho_R=\tr_A\ket{\chi}\bra{\chi}$, then for any pair of weight matrices $\phi_A$ on $A$ and $\phi_R$ on $R$ such that
\beq
\langle e^{A}_i|\phi_A| e^{A}_i\rangle =\langle e^{R}_i|\phi_R| e^{R}_i\rangle \quad \hbox{$\forall\;\;i$ with }\;0<\lam_i<1 \nn
\eeq
one has 
\beq
S_{\phi_1}(\rho_A)= S_{\phi_R}(\rho_R).\nn
\eeq
In this case we say that $\phi_A$ is $(\rho_A,\rho_R)$-conjugated to $\phi_R$ and 
$\phi_R$ is $(\rho_R,\rho_A)$-conjugated to $\phi_A$.
\end{proposition}
From this one has the following corollary:
\begin{corollary}\label{cor1}
Let $\rho$ be  a density matrix in $AB$, as well as
$\rho_A=\tr_B \rho$ and  $\rho_B=\tr_A\rho$.
Take a reference system $R$ with Hilbert space isomorphic to the Hilbert space of $AB$,
and a pure state $\ket{\chi}$ on $ABR$ such that $\rho =\tr_R\ket{\chi}\bra{\chi}$. Set $\rho_R= 
\tr_{AB} \ket{\chi}\bra{\chi}$ and
$\rho_{BR}=\tr_{A}  \ket{\chi}\bra{\chi} $, 
then 
\beq S_\phi(\rho )=S_{\phi_R}(\rho_R)\;\hbox{ and }\;S_{\phi_A}(\rho_A)=S_{\phi_{BR}}(\rho_{BR}) 
\eeq
if $\phi$ is $(\rho,\rho_R)$-conjugate to $\phi_R$ and $\phi_A$ is $(\rho_A,\rho_{BR})$-conjugate to $\phi_{BR}$.
\end{corollary}

\section{A diagonalisation bound}

Another simple trace inequality, which is also based on the weighted Gibbs inequality, deals with the projection of the density matrix on its diagonal in a given basis, as occurs for example when projective measurements are performed. Let $\rho$ be a density matrix
in $\mathcal{H}$. 
Let $\ket{f_1},\ldots ,\ket{f_d}$ be a basis in $\mathcal{H}$ and $\rho^{\rm d}$ denote the diagonal part of 
$\rho$ in this basis, i.e. $\langle f_j| \rho^{\rm d}| f_k\rangle =\delta_{jk}\langle f_j|\rho | f_j\rangle$ for $1\leq j,k\leq d$.
Then we have the following bound for the weighted entropy of $\rho^{\rm d}$

\begin{theorem} \label{thm:diag} Under the condition $\tr\phi\rho \geq \tr \phi\rho^{\rm d}$
\beq
S_\psi (\rho^{\rm d})\geq S_{\phi} (\rho ),
\eeq
with equality if and only if $\rho=\rho^{\rm d}$ and where $\psi$ fulfils $\langle f_j |\psi 
\rho^{\rm d}  | f_j\rangle=\langle f_j |\phi \rho  | f_j\rangle$.
\end{theorem}

\begin{proof}
The proof is again an application of the weighted Gibbs inequality, Theorem \ref{Gibbs}. Choosing $\sigma=\rho^{\rm d}$, by the condition of the theorem the Gibbs inequality can be used, yielding
\bea
0&\leq& D_\phi(\rho \| \rho^{\rm d} )= -S_{\phi} (\rho )  - \tr \phi \rho \log \rho^{\rm d}   \nonumber\\
&=& -S_{\phi} (\rho )  - \sum_j  \langle f_j |\phi \rho  | f_j\rangle \log  \langle f_j |\rho^{\rm d} | f_j\rangle=-S_{\phi} (\rho ) +S_\psi (\rho^{\rm d})
\eea
with inequality for $\rho=\rho^{\rm d}$. This completes the proof.
\end{proof}

Note that in the special case of von Neumann entropy, one has $\phi=\psi=1_{\mathcal{H}}$ and all conditions of the theorem are automatically satisfied.

 \section{Subadditivity of quantum weighted entropy}
Let us first focus on a composite system $AB$ of two components $A$ and $B$ with density matrix $\rho_{AB}$ and weight $\phi_{AB}=\phi_A\otimes\phi_B$. Recall the standard reduced density matrices defined by taking the partial trace, i.e. $\rho_{A} =\tr_B (\rho_{AB})$ and so on. We can now prove the following subadditivity property of quantum weighted entropy:

\begin{theorem}[Subadditivity] \label{thm:sub}
Under the condition $\tr_{AB}(\phi_{AB} \rho_{AB})\geq \tr_{A}(\phi_{A} \rho_{A}) \tr_{B}(\phi_{B} \rho_{B})$ one has 
\beq
S_{\phi_{AB}}(\rho_{AB})\leq S_{\psi_{A}}(\rho_{A})+S_{\psi_{B}}(\rho_{B})
\eeq
with equality for $\rho_{AB}=\rho_A\otimes\rho_B$, where the reduced weights are defined implicitly through $\psi_{A} \rho_{A}=\tr_B(\phi_{AB} \rho_{AB})$ and similarly for $B$.
\end{theorem}

\begin{proof} 
The condition stated in Theorem \ref{thm:sub}, i.e. $\tr_{AB}(\phi_{AB} \rho_{AB})\geq \tr_{A}(\phi_{A} \rho_{A}) \tr_{B}(\phi_{B} \rho_{B})$, ensures that we can use the weighted Gibbs inequality with $\sigma_{AB} =\rho_A\otimes\rho_B$ and $\phi =\phi_{AB}=\phi_A \otimes\phi_B$. Further abbreviating $\rho=\rho_{AB}$, one gets  
\beq
D_\phi(\rho \| \rho_A\otimes\rho_B )\geq 0  
\eeq
with equality if and only if $\rho =\rho_A\otimes\rho_B $. Simplifying the above gives
\bea
0&\leq& D_\phi(\rho \| \rho_A\otimes\rho_B ) \nonumber\\
&=&\tr_{AB} \left\{ \phi \rho \left( \log\,\rho - \log\,(\rho_A\otimes\rho_B) \right)\right\}
\eea
Hence,
\bea
S_\phi (\rho ) \! &\leq& \!\!\! - \tr_{AB} \left\{ \phi\rho \left[ \log\,(\rho_A\otimes 1_B)- \log\,(1_A\otimes\rho_B)\right]\right\} \nonumber \\
&=&  \!\!\!  - \tr_{A}\{ \tr_B( \phi\rho) \log\rho_A  \} \!-\! \tr_{B}\{ \tr_A( \phi\rho  ) \log\rho_B \} \nonumber \\
&=& S_{\psi_{A}}(\rho_{A})+S_{\psi_{B}}(\rho_{B})
\eea
under the above definition of reduced weights. This completes the proof.
\end{proof}

\section{Concavity of quantum weighted entropy}

\def\cK{\mathcal K}  

We can use the subadditivity property of quantum weighted entropy proved in the previous section to show that quantum weighted entropy is concave in analogy to the case of standard von Neumann entropy. 
%
%
\begin{theorem} \label{ThmConcavity}
Suppose that $\rho^{(1)},\ldots ,\rho^{(r)}$ are density matrices of a system $A$ a Hilbert space $\mathcal{H}_A$ and $\bfb =(b_1,\ldots b_r)$ is 
a probability vector, with non-negative entries and $\sum_{1\leq l\leq r} b_l =1$. 
Then
\beq
S_\phi\left(\sum_lb_l\rho^{(l)} \right) \geq\sum_l b_l S_\phi(\rho^{(l)}) \label{concave-thm-eq}
\eeq
with equality if and only if $b_l=1$ for some $l$ or $\rho^{(l)}=\rho^{(1)}$ $\forall$ $l$.
\end{theorem}

%
%


\begin{proof} 
Set
\beq
\sigma=\sum\limits_lb_l\rho^{(l)},\quad\!\! {\rB}_l=\tr\phi\rho^{(l)},\quad\!\! {\rA}=\tr\phi\sigma =\sum_lb_l\rB_l.
\eeq
Recall the expressions for the standard Shannon entropy of $\bfb$ and the weighted Shannon entropy of $\bfb$ with weight $B$,
\beq
h(\bfb )=-\sum_lb_l\log\,b_l, \quad h_{\rB} (\bfb ) =-\sum_l{\rB}_l  b_l\log\,b_l. 
\eeq
Take an auxiliary system $R$ with Hilbert space $\mathcal{H}_R$ of dimension $r$
and fix a basis $\ket{e_1} ,\ldots,\ket{e_r} $ in  $\mathcal{H}_R$. Consider a density
matrix $\rho$ on $AR$ defined by the condition
that for all $\ket{v},\ket{v'}\in \mathcal{H}_A$ and $1\leq l,l^\prime\leq r$:
\beq
\big\langle v\otimes e_l|\rho | v'\otimes e_{l'}\big\rangle =
b_l\langle v |\rho^{(l)}|v'\rangle\delta_{l,l'}. \label{conc-rhodef}
\eeq
It is easily verified that $\rho$ is indeed a density matrix, i.e.\ a positive-definite operator
of trace 1.
Then 
\beq
\rho_A =\tr_R \rho =\sum_l b_l\rho^{(l)}=\sigma 
\eeq 
and
$\rho_R=\tr_A \rho$ is diagonal in basis 
$\ket{e_1} ,\ldots,\ket{e_r}$, with diagonal entries $b_1,\ldots 
b_r$. Also, if $\rho^{(l)}$ has eigenvectors $\ket{e^{(l)}_j}$ with eigenvalues
$\lambda^{(l)}_j$ then $\rho$ has the eigenvectors $\ket{e^{(l)}_j}\otimes
\ket{e_l} $ with the eigenvalues $\lambda^{(l)}_jb_l$. 
Hence, with $1_{R}$ denoting the unit operator on $R$ one has
\bea
S_{\phi\otimes 1_R}(\rho)&=&-\sum_{j,l}
\big\langle e^{(l)}_j\big|\phi\big| e^{(l)}_j\big\rangle(\lambda^{(l)}_jb_l)\log (\lambda^{(l)}_jb_l)  \nn\\
&=&-\sum_{l}b_l\sum_j\big\langle e^{(l)}_j\big|\phi\big| e^{(l)}_j\big\rangle\lambda^{(l)}_j
\log\lambda^{(l)}_j 
-\sum_l{\rB}_lb_l\log b_l     \nn\\
&=&\sum_l b_l S_\phi(\rho^{(l)})+ h_{\rB} (\bfb ). \label{Concave-eq-S}
\eea
Note that $\tr(\phi\otimes 1_R)\rho=\tr(\phi\rho_A)$, so the bound $\tr(\phi\otimes 1_R)\rho\geq \tr(\phi\otimes 1_R)(\rho_A\otimes\rho_R)$ is fulfilled. 
Finally, by \eqref{conc-rhodef} the partial trace  $T=\tr_A(\phi\otimes 1_R)\rho$ is a diagonal matrix 
in the basis $\ket{e_1} ,\ldots,\ket{e_r} $ of $R$, with entries 
\beq
\langle e_l|  T | e_{l'}\rangle =\delta_{l,l'}b_l\,\rB_l. 
\eeq
To complete the proof we use the subadditivity property proven in the previous section in Theorem \ref{thm:sub} for the joint system $AR$ with density matrix $\rho$ and weight $\phi\otimes 1_R$. Therefore, we introduce reduced weights defined implicitly through 
\bea
\psi_{A} \rho_{A}&=&\tr_R (\phi\otimes 1_R) \rho = \phi \sigma \nn \\
\psi_{R} \rho_{R}&=& \tr_A (\phi\otimes 1_R) \rho 
=T= \sum\limits_l b_l \, {\rB}_l    | e_l\rangle\langle e_l|.     \nn
\eea
Then one has
\beq
S_{\psi_A}(\rho_A)=S_\phi \left(\sigma\right),\;\hbox{ and }\;
S_{\psi_R}(\rho_R)=h_{\rB}(\bfb ) .
\eeq 
Therefore, 
subadditivity (Theorem \ref{thm:sub}) yields
\beq
S_{\phi\otimes 1_R}(\rho)\leq S_{\psi_A}(\rho_A)+S_{\psi_R}(\rho_R)
=S_\phi \left(\sigma\right)+h_{\rB}(\bfb)
\eeq
with equality if and only if $\rho =\rho_A\otimes\rho_R,$ i.e. $h(\bfb )=0$ or $\rho^{(l)}=\rho^{(1)}$ $\forall$ $l$. This together with \eqref{Concave-eq-S} gives \eqref{concave-thm-eq}. 
\end{proof}

\section{Araki-Lieb inequality for quantum weighted entropy}
Consider a composite system $AB$ with density matrix $\rho$, weight $\phi$ and partial density matrices $\rho_A=\tr_B \rho$ and  $\rho_B=\tr_A\rho$.
Construct a purification by introducing a reference system $R$ with Hilbert space isomorphic to the Hilbert space of $AB$,
and a pure state $\ket{\chi}$ on $ABR$ such that $\rho =\tr_R\ket{\chi}\bra{\chi}$. Furthermore, set $\rho_R= \tr_{AB} \ket{\chi}\bra{\chi}$ and $\rho_{BR}=\tr_{A}  \ket{\chi}\bra{\chi} $, then by Corollary \ref{cor1} we have 
\beq \label{AL1}
S_\phi(\rho )=S_{\phi_R}(\rho_R)\;\hbox{ and }\;S_{\phi_A}(\rho_A)=S_{\phi_{BR}}(\rho_{BR}) 
\eeq
if $\phi$ is $(\rho,\rho_R)$-conjugate to $\phi_R$ and $\phi_A$ is $(\rho_A,\rho_{BR})$-conjugate to $\phi_{BR}$. Combining this with subadditivity gives rise to the following result.

\begin{theorem}[Weighted Araki-Leib inequality] \label{thm:AL} One has
\beq
S_\phi(\rho)\geq \Big(\sup_{\Psi\in\mathcal{D}} \big[S_{\psi_A}(\rho_A)-S_{\psi_B}(\rho_B)\big]\Big)\vee
 \Big(\sup_{\bar{\Psi}\in\bar{\mathcal{D}}} \big[S_{\opsi_B}(\rho_B)-S_{\opsi_A}(\rho_A)\big]\Big),
\eeq
where the set $\mathcal{D}(\phi)$ consists of all pairs $\Psi=(\psi_A,\psi_B)$ for which there exists a $\phi_{BR}$ and $\psi_{R}^*$ implicitly defined through $\psi_{R}^*\rho_R=\tr_B \phi_{BR}\rho_{BR}$ satisfying $\tr_{BR} \phi_{BR}\rho_{BR}\geq \tr_{BR} \phi_{BR}\rho_{B}\otimes\rho_{R}$ and $\psi_{R}^*$ is $(\rho_R,\rho)$-conjugate to $\phi$, such that $\psi_A$ is $(\rho_A,\rho_{BR})$-conjugate to $\phi_{BR}$ and $\psi_B$ is $\rho_B$-conjugate to $\psi_{B}^*$ defined through $\psi_{B}^*\rho_B=\tr_R \phi_{BR}\rho_{BR}$; and similarly for  $\bar{\mathcal{D}}(\phi)$.
\end{theorem}

\begin{proof} To proof the theorem is suffices to show that 
\beq\label{AL2}
S_\phi(\rho)\geq S_{\psi_A}(\rho_A)-S_{\psi_B}(\rho_B), \quad \mathrm{and}\quad S_\phi(\rho)\geq S_{\opsi_B}(\rho_B)-S_{\opsi_A}(\rho_A)
\eeq
for all $(\psi_A,\psi_B)\in\mathcal{D}(\phi)$ and $(\opsi_A,\opsi_B)\in\bar{\mathcal{D}}(\phi)$. We start by establishing the first inequality in \eqref{AL2}. For any $(\psi_A,\psi_B)\in\mathcal{D}(\phi)$ there exists a $\phi_{BR}$, $\psi_{B}^*$ and $\psi_{R}^*$ with $\psi_{B}^*\rho_B=\tr_R \phi_{BR}\rho_{BR}$ and $\psi_{R}^*\rho_R=\tr_B \phi_{BR}\rho_{BR}$, satisfying $\tr_{BR} \phi_{BR}\rho_{BR}\geq \tr_{BR} \phi_{BR}\rho_{B}\otimes\rho_{R}$. We can thus apply the subadditivity inequality, Theorem \ref{thm:sub}, to obtain
\beq
S_{\phi_{BR}} (\rho_{BR})\leq S_{\phi_{B}^*} (\rho_{B})+S_{\phi_{R}^*} (\rho_{R})
\eeq
Since $\psi_A$ and $\phi_{BR}$ are $(\rho_A,\rho_{BR})$-conjugate, one has $S_{\phi_{BR}} (\rho_{BR})=S_{\psi_A}(\rho_A)$. Further, since $\psi_{B}^*$ and $\psi_B$ are $\rho_B$-conjugate, one has $S_{\phi_{B}^*} (\rho_{B})=S_{\psi_B}(\rho_B)$. Moreover, the condition that $\psi_{R}^*$ is $(\rho_R,\rho)$-conjugate to $\phi$ implies $S_{\phi_{R}^*} (\rho_{R}) =S_\phi(\rho)$. This proves the first inequality in \eqref{AL2}. The second inequality in \eqref{AL2} is established in a similar manner. This completes the proof.
\end{proof}

\begin{remark} Note that in the case of von Neumann entropy with $\phi=1_{AB}$, one has $\psi_A=\opsi_A=1_A$ and $\psi_B=\opsi_B=1_B$ and thus \eqref{AL2} simply reads
\beq
S(\rho)\geq |S(\rho_A)-S(\rho_B)|
\eeq 
which is the Araki-Leib inequality for von Neumann entropy.
\end{remark}

\section{Strong subadditivity of quantum weighted entropy}

Another very interesting trace inequality concerns the strong subadditivity property for a composite system $ABC$ with density matrix $\rho_{ABC}$ and weight $\phi_{ABC}=\phi_A\otimes\phi_B\otimes\phi_C$.
\begin{theorem}[Strong subadditivity]  \label{thm:strong}
Under the conditions $(i)$ $\tr_{ABC}(\phi_{ABC} \rho_{ABC})\geq \tr_B\left\{\phi_B \tr_{A}(\phi_{A} \rho_{AB}) \tr_{C}(\phi_{C} \rho_{BC}) \rho_B^{-1} \right\}$, as well as $(ii)$ $[\rho_{AB},\phi_A\otimes \phi_B]=0$ and $[\tr_C(\phi_C \rho_{BC}),\rho_B]=0$ one has
\beq 
S_{\phi_{ABC}}(\rho_{ABC})\!+\!S_{\psi_{B}}(\rho_{B}) \! \leq \! S_{\psi_{AB}}(\rho_{AB}) \! +\! S_{\psi_{BC}}(\rho_{BC})
\eeq
where the reduced weights are defined as above, i.e. $\psi_{AB} \rho_{AB}=\tr_C(\phi_{ABC} \rho_{ABC})$ and so on.
\end{theorem}

\begin{remark} Let us make two remarks regarding the conditions of the theorem. Firstly, as expected, both conditions are automatically satisfied in case of $\phi=1_{ABC}$, i.e.\ the case of standard von Neumann entropy. Secondly, condition $(i)$ is the natural analog of the condition of the subadditivity property and as in the latter case is a physical condition which one expects not to be able to improve on. However, condition $(ii)$ is a technical condition which one might hope to improve. Further, note that an analogs condition with $A$ and $C$ interchanged is also a valid condition $(ii)$.
\end{remark}

\begin{proof}[Proof of Theorem \ref{thm:strong}]
To prove Theorem \ref{thm:strong} we have to show that $\mathcal{A}:=S_{\phi_{ABC}}(\rho_{ABC})\!+\!S_{\psi_{B}}(\rho_{B}) - S_{\psi_{AB}}(\rho_{AB}) \!  - \! S_{\psi_{BC}}(\rho_{BC}) \leq 0$. Since we have already proven the weighted Klein's inequality we can follow a similar strategy as in the original proof by Lieb and Ruskai \cite{LR1} of strong subadditivity of standard von Neumann entropy which uses the standard Klein's inequality as a first ingredient. To do so we first make the following observation, which follows from the specific definition of the reduced weights as given in Theorem \ref{thm:strong}; abbreviating $\rho\equiv \rho_{ABC}$ and $\phi\equiv \phi_{ABC}$, one gets
\bea
\mathcal{A}=\! \tr_{ABC} \! \left\{  \phi\rho \left( \log\rho_{AB} +\log\rho_{BC}-\log\rho_B -\log\rho  \right) \right\},
\eea
 where all matrices are to be understood as extended to the Hilbert space of the full system $ABC$, i.e. $\rho_{AB}$ is short for $\rho_{AB}\otimes 1_C$ and similarly for the others. Now we can apply the weighted Klein's inequality, Lemma \ref{Klein}, with $W=\phi$, $Y=\rho$ and $X=\exp(\log\rho_{AB}-\log\rho_B+\log\rho_{BC})$ and $f(x)=x\log x$ as in \eqref{Kleineq2}, yielding 
\bea \label{eqinter}
\mathcal{A}\leq  \tr_{ABC} \! \left\{  \phi \exp(\log\rho_{AB}\! -\! \log\rho_B\!+\!\log\rho_{BC})\! -\! \phi \rho \right\}
\eea 
This relation very much resembles the Golden-Thomson inequality \cite{GT1,GT2}
\beq
\tr \left( e^{X+Y}\right)\leq \tr \left(e^{X} e^{Y}\right)
\eeq
and in the proof of the strong subadditivity for standard von Neumann entropy and Lieb and Ruskai \cite{LR1} used a generalisation of the Golden-Thomson inequality derived in an earlier work by Lieb \cite{Lieb},
\beq \label{GTLieb}
\tr \left( e^{X+Y+Z} \right )\leq \tr \left( e^Z T_{\exp(-X)}(e^Y)  \right)
\eeq
where
\beq
T_{\exp(-X)}(e^Y) = \int_0^\infty  (e^{-X} + 1\omega)^{-1} e^Y  (e^{-X} + 1\omega)^{-1} d\omega.
\eeq 
This relation can be extended to the weighted case, i.e. for $W$ being a weight one has
\beq \label{newineq}
\tr \left(W e^{X+Y+Z} \right )\leq \tr \left(K_W(Z)\, T_{\exp(-X)}(e^Y)  \right).
\eeq 
with 
\beq
K_W (Z) = \sum_{n=0}^\infty \frac{1}{(n+1)!} \sum_{l=0}^n Z^{n-l} W Z^l  
\eeq
and $T_{\exp(-X)}(e^Y) $ as defined above.
The proof is a generalisation of the proof of Theorem 7 of \cite{Lieb}. Set $\xi=e^{-X}$, $\eta=e^Y$ and $R=X+Z$. Note that $\xi$ and $\eta$ are strictly positive operators. We define a function $F$ from the cone of strictly positive operators to the real numbers through $F: \xi\to-\tr[ W \exp(R+\log\xi)]$. Since $F$ is convex and homogeneous of order one, Lemma 5 of \cite{Lieb} can be applied, yielding
\beq
-\tr \left(W e^{X+Y+Z} \right ) = F(\eta)\geq \frac{d}{d\omega} \left[F(\xi+\omega \eta)\right]_{\omega=0}.
\eeq
Taylor expanding and taking the derivative gives \eqref{newineq}.

We can now apply \eqref{newineq} to the first term on the right-hand-side of \eqref{eqinter}. Choosing $X \!=\!-\! \log\rho_B$, $Y\!=\! \log\rho_{BC}$, $Z=\!\log\rho_{AB}$ and $W=\phi$ and furthermore assuming under condition $(ii)$ that we have the commutation relations $[\rho_{AB},\phi_A\otimes \phi_B]=0$ and $[\tr_C(\phi_C \rho_{BC}),\rho_B]=0$, then the first term on the right-hand-side of \eqref{eqinter} is bounded from above by
 \bea
&& \tr_{B} \Big\{  \phi_B\,  \tr_A(\phi_A \rho_{AB})  \tr_C(\phi_C \rho_{BC}) 
\int_0^\infty (\rho_B + 1\omega)^{-1} (\rho_B  + 1\omega)^{-1} d\omega \Big\} \nn\\
&&=\tr_{B} \Big\{  \phi_B\,  \tr_A(\phi_A \rho_{AB})  \tr_C(\phi_C \rho_{BC}) \rho_B^{-1}  \Big\}
 \eea
Thus one arrives at
\beq
\mathcal{A}  \leq  \tr_{B} \Big\{  \phi_B\,  \tr_A(\phi_A \rho_{AB})  \tr_C(\phi_C \rho_{BC}) \rho_B^{-1}  \Big\} -    \tr_{ABC}  \phi  \rho 
\eeq
which by condition $(i)$ of the theorem yields
\beq
\mathcal{A}  \leq 0.
\eeq
This completes the proof.
\end{proof}

\section{Discussion}

We introduce quantum weighted entropy and derived several useful properties in terms of various trace inequalities. Each of those inequalities contains the corresponding result for von Neumann entropy as a special case when the weight is chose to be the identity matrix. In particular, besides basic properties, we derived a diagonalisation bound (Theorem \ref{thm:diag}), subadditivity and concavity of quantum weighted entropy (Theorem \ref{thm:sub} and \ref{ThmConcavity}), an analog of the Araki-Lieb inequality (Theorem \ref{thm:AL}) and strong subadditivity of quantum weighted entropy (Theorem \ref{thm:strong}). An essential ingredient to the previous results is an analog of Gibbs inequality for quantum weighted relative entropy (Theorem \ref{Gibbs}) which in turn is obtained from a weighted Klein's inequality (Lemma \ref{Klein}).  

A difficulty in proving the above trace inequalities, in comparison to the analogous results for von Neumann entropy, is the fact that in general the weights do not commute with the (reduced) density matrices. In the case of the weighted Klein's inequality we circumvent this problem by utilising the unique decompositions $W=L L^\dag$ of the weight. Since the weighted Gibbs inequality and in turn most of the other inequalities are derived from the weighted Klein's inequality, they inherited this property and can be proven without any further assumptions on commutation relations of the weight. The only result, where commutativity of parts of the weight with some of the reduced density matrices is assumed, is strong subadditivity. The proof is thus not optimal and one would hope to be able to improve it by relaxing those conditions. In this context we note that our proof of strong subadditivity of quantum weighted entropy closely follows the original proof by Lieb and Ruskai for strong subadditivity of von Neumann entropy \cite{LR1}. This enables one to use the weighted Gibbs inequality in an essential manner. In the case of alternative, more modern strategies for proving strong subadditivity, as in \cite{NP}, the situation is more involved. 

The here presented discussion of quantum weighted entropy is a first account deriving many of its properties and thus forms a basis for further interesting potential applications of the latter in quantum information theory. 
\\

{\emph{Acknowledgements --}}
YS and SZ thank Salimeh Yasaei Sekeh for useful discussions.
YS thanks University of Sao Paulo (at Sao Paulo and at Sao Carlos) for the hospitality during the academic year 2013-4.
SZ acknowledges support by CNPq (Grant 307700/2012-7) and PUC-Rio, as well as thanks USP for kind hospitality.

\appendix
\section{Proof of weighted Klein's inequality}
In case the matrices $X$ and $W$ commute, one can simultaneously diagonalise them, which enables one to follow the same steps as in the proof for the standard Klein's inequality. The general case, where $X$ and $W$ no \emph{not} commute, is slightly more involved and relies on the following decomposition: Since $W$ is positive definite one has that $\langle v| W | v \rangle\geq0$ for any $|v\rangle$ and there exists a unique $L$ such that $W=L L^\dag$. Let now $|e_1\rangle,|e_2\rangle,\ldots$, be the normalised eigenvectors of $X$ and $\la_1,\la_2,\ldots$ the corresponding eigenvalues.  Furthermore, we define the \emph{normalised} vectors $| \tilde{e}_j \rangle=  L | e_j \rangle / \sqrt{\langle e_j| W | e_j \rangle}$. Then
\bea
\!\!\!\!\! \!\!\!\!\!  \!\!\!\!\! && \tr \left( W(f(Y) - f(X))\right) = \nonumber \\
&&=\sum_j \langle e_j| W | e_j \rangle \left\{ \langle \tilde{e}_j| f(Y)  | \tilde{e}_j \rangle -f(\la_j)   \right\} 
\eea
Now we use that for any unit vector $|v \rangle \in{\cal H}$, by convexity of $f$, 
\beq
\langle v |f( Y)| v\rangle \geq {f\big(\langle v | Y| v \rangle\big)}. 
\eeq
Also, $f(y)-f(x)\geq (y-x)f'(x)$ for $x,y\in{\mathbb R}$. Thus
\bea 
\!\!\!\!\! \!\!\!\!\!  \!\!\!\!\!  \!\!\!\!\! && \tr \left(W(f(Y) - f(X))\right)  \nonumber\\
&& \geq \sum_j \langle e_j| W | e_j \rangle \left\{ f\left( \langle \tilde{e}_j | Y | \tilde{e}_j \rangle\right) -f(\la_j)   \right\} \nonumber\\
&&  \geq \sum_j \langle e_j| W | e_j \rangle \left\{ \langle \tilde{e}_j | Y | \tilde{e}_j \rangle - \la_j   \right]  f'(\la_j) \nonumber\\
&&= \tr \left( W(Y -X) f'(X) \right), 
\eea
which completes the proof.


\section*{References}
 \begin{thebibliography}{10}
\providecommand*{\bibinfo}[2]{#2}
\providecommand*{\epfmt}[2]{#2}
\providecommand*{\eprint}[1]{#1}
\providecommand*{\url}[1]{#1}

\bibitem{CT} T. Cover and J. Thomas. {\it Elements of Information Theory.} Wiley, New York, 2006.

\bibitem{Information} M. Kelbert and Y. Suhov. {\it Information Theory and Coding by Example.} Cambridge University Press, Cambridge, 2013.

\bibitem{Nielsen}
M.~A. Nielsen and I.~L. Chuang,
{\it Quantum Computation and Quantum Information}
Cambridge University Press, Cambridge, 2000.

\bibitem{BG} M. Belis and S. Guiasu,  
{\it IEEE Trans. on Information Theory},  \textbf{14} (1968), 593--594.

\bibitem{G} S. Guiasu,
{\it Report on Math. Physics}, \textbf{2} (1971), 165--179.


\bibitem{ShMM} B. D. Sharma, J. Mitter and M. Mohan. 
{\it Inform. Control} \textbf{39} (1978), 323--336.

\bibitem{SiB} R. P. Singh and J. D. Bhardwaj. 
{\it Inf. Sci.} \textbf{59} (1992), 149--163.

\bibitem{K} J. N. Kapur. 
{\it Measures of Information and Their Applications.} Chapter 17,
New Delhi: Wiley Eastern Limited, 1994.

\bibitem{SrV} A. Sreevally and S. K. Varma,
{\it Soochow Journal of Mathematics}, \textbf{30} (2004), no. 2, 237--243.

\bibitem{S} A. Srivastava, 
{\it Cybernetics and Information Technologies}, \textbf{11} (2011), no. 3, 60--65.

\bibitem{MMN} K. Muandet, S. Marukatat and C. Nattee,
in {\it Advances in Machine Learning.} Lecture Notes in Computer Science, \textbf{5828} (2009), 278--292.

\bibitem{LR}
O. Lanford and D. W. Robinson, 
{\it J. Math. Phys.} {\bf 9} (1968), 1120.

\bibitem{LR1}
E. H. Lieb and M. B. Ruskai,
{\it J. Math. Phys.} {\bf 14} (1973), 1938.

\bibitem{LR2}
E. H. Lieb and M. B. Ruskai,
{\it Phys. Rev. Lett.} {\bf 30} (1973), 434.

\bibitem{NP}
M.~A. Nielsen and D.~Petz,
{\it Quantum Information and Computation} {\bf 5} (2005), 507--513. 

\bibitem{RR}
D. W. Robinson and D. Ruelle, 
{\it Commun. Math. Phys.} {\bf 5} (1967), 288.

\bibitem{Ruskai2002} M.~B.~Ruskai,
{\it J. Math. Phys.} {\bf 43} (2002), 4358.

\bibitem{SY} Y. Suhov and S. Yasaei Sekeh, 
{\it arxiv:} 1409.4102 (2014).

\bibitem{GT1}
S. Golden, 
{\it Phys. Rev. B} {\bf 137} (1965), 1127-1128.
1128.
 
\bibitem{GT2}
C. J. Thompson, 
{\it J. Math. Phys.} {\bf 6}  (1965), 1812-1813.

\bibitem{Lieb}
E. H. Lieb,
{\it Adv. in Math.} {\bf 11} (1973) 267--288.

\end{thebibliography}
\end{document}